\RequirePackage{fix-cm}
\documentclass[smallextended]{svjour3}       
\smartqed  
\usepackage{graphicx}
\usepackage{url}
\usepackage[utf8x]{inputenc} 
\usepackage{amsmath}
\usepackage{mathrsfs}
\usepackage{amsfonts}
\usepackage{amssymb}

\usepackage{titletoc}
\usepackage{lmodern}
\usepackage[T1]{fontenc}
\usepackage{textcomp}
\usepackage{accents}
\usepackage[mathscr]{euscript}
\usepackage{graphicx}
\usepackage[misc,geometry]{ifsym}
\newtheorem{ass}{Assumption}

\newcommand{\MAP}[3]{{#1}:{#2}\to\mathbb{R}^{#3}}

\newcommand{\rea}{\mathbb{R}}
\newcommand{\Rr}{\mathbb{R}}

\usepackage{color}

\begin{document}

\title{Contraction-based Tracking Control of Electromechanical Systems
}


\author{Najmeh Javanmardi        \and
         Pablo Borja \and
         Jacquelien M. A. Scherpen \Letter       
}


\institute{N. Javanmardi and J. M. A. Scherpen  \at
              Jan C. Willems Center for Systems and Control, ENTEG, Faculty of Science and Engineering, University of Groningen, Groningen, The Netherlands \\
              \email{\{n.javanmardi, j.m.a.scherpen\}@rug.nl}          
           \and
           P. Borja \at
School of Engineering, Computing and Mathematics, University of Plymouth, Plymouth, United Kingdom \\
\email{pablo.borjarosales@plymouth.ac.uk}
}

\date{Received: date / Accepted: date}

\maketitle

\begin{abstract}
This paper addresses the trajectory-tracking problem for a class of electromechanical systems. To this end, the dynamics of the plants are modeled in the so-called port-Hamiltonian framework. Then, the notion of contraction is exploited to design the desired closed-loop dynamics and the corresponding tracking controller.
Notably, the proposed control design method does not require solving partial differential equations or changing the coordinates of the plant, which permits preserving the physical interpretation of the controller.
The applicability of the proposed approach is illustrated in several electromechanical systems via simulations.
\keywords{Electromechanical systems \and Port-Hamiltonian systems \and Trajectory tracking \and Energy shaping \and Contractive systems.}

\end{abstract}

\section{Introduction}
\label{intro}
Electromechanical (EM) technology has been increasingly used in numerous applications, e.g., electric drives, motors, micro and nanoelectromechanical systems (MEMS/NEMS), microphones, and switches \cite{dean2009applications,zhang2014electrostatic}. EM systems consist of two subsystems---one mechanical and one electrical---coupled through nonlinear dynamics. Concerning modeling approaches for EM systems, the so-called port-Hamiltonian (pH) framework (see \cite{GEO,VANJEL}) has proven suitable to represent the multi-domain and nonlinear nature of these systems.

Due to the nonlinear couplings between the mechanical and electrical subsystems, inherent instability issues may arise in EM systems. Such issues result in operational limitations on stabilization problems, e.g., pull-in instability in MEMS devices \cite{zhang2014electrostatic}. In the literature, 
 operational range enlargement (e.g.,\cite{zhang2015continuous}) and the stabilization (e.g., \cite{borovic2004control,chu1994analysis,maithripala2005control,nadal,nunna2015constructive,rodriguez2003stabilization}) and trajectory-tracking problems (e.g., \cite{OWUSU,rodriguez2021exponential,sun,zhang2015continuous,Zhu}) for EM systems have been studied by adopting diverse control approaches. Among such strategies, energy-shaping methods have proven suitable for controlling EM systems while preserving the physical interpretation of the closed-loop system. Notably, energy-shaping control methods focus on modifying the energy of the closed-loop system to guarantee the stability of the desired equilibrium or trajectory. In this regard, the interconnection and damping assignment passivity-based control (IDA-PBC) methodology \cite{ortega2002interconnection} is the most versatile energy shaping method in terms of the plants that can be stabilized, as it has been shown in a wide range of applications, e.g., \cite{ortega2001putting,rodriguez2003stabilization}. As part of the IDA-PBC approach, it is essential to solve a set of partial differential equations (PDEs), referred to as matching equations. However, finding the solution to these PDEs remains the most significant bottleneck in IDA-PBC.

%
%

 In \cite{yaghmaei2017trajectory}, the authors propose an IDA-PBC version to solve the trajectory-tracking problem for a class of pH systems. Such an approach is named timed interconnection and damping assignment passivity-based control (tIDA-PBC). This method is based on the notion of contractive pH systems. In particular, in exploiting the fact that the trajectories of contractive systems converge exponentially to each other. Consequently, the proposed tracking method ensures that the trajectories of the closed-loop system converge exponentially to the desired trajectories. An application of this approach is given in the recent paper \cite{rodriguez2021exponential}, where the author implements the tIDA-PBC method in a permanent magnet synchronous motor.

The main contribution of this work consists in proposing a contraction-based method to solve the trajectory-tracking problem for a large class of EM systems. Additionally, the proposed control design approach does not require solving PDEs or any change of coordinates. We stress that methods for controlling EM systems without solving PDEs---e.g., \cite{BOR,nunna2015constructive}---are only suitable for stabilization purposes. On the other hand, methods that solve the trajectory-tracking problem for EM systems---e.g., \cite{rodriguez2021exponential}---require solving PDEs, which may hamper their implementation.


The remainder of the paper is organized as follows: Section \ref{sec:EM moddel} introduces the class of EM systems considered in this paper and provides the problem formulation. The main result of this work, i.e., the control approach to solving the trajectory-tracking problem for EM systems, is given in Section \ref{sec:control}. The applicability of the proposed control method is illustrated in Section \ref{sec:ex}. Finally, some concluding remarks are given in Section \ref{sec:con}.

{\it Notation.} Given a square matrix $A\in\rea^{n\times n}$, the symbols $A\succ0$ and $A\succeq0$ denote that the matrix $A$ is positive definite and positive semi-definite, respectively. Given a differentiable function $H:\rea^{n}\to \rea$, $\nabla H(x)$ is defined as $\nabla H(x)\triangleq[\,\frac{\partial H(x)}{\partial x_1}, \frac{\partial H(x)}{\partial x_2}, \dots, \frac{\partial H(x)}{\partial x_n}]\,^\top$, for $x\in\rea^{n}$. Similarly, ${\nabla}^2 H(x)$ is a matrix whose $ij$th element is given by $\frac{{\partial}^2 H(x)}{\partial x_i \partial x_j}$.
For a full-rank matrix $g \in \Rr^{n \times m}$, with $n\geq m$, we define $g^{\dagger}\triangleq(g^{\top}g)^{-1}g^{\top}$.

 \section{Modeling and Problem Formulation }
\label{sec:EM moddel}

In this section, we present a pH model suitable for representing a broad range of EM systems. This model consists of a mechanical subsystem and an electrical subsystem. The states of the mechanical part are represented by the generalized positions $q\in\rea^{n_{\tt m}}$ and momenta $p\in\rea^{n_{\tt m}}$, and the state of the electrical subsystem is given by $x_{\tt e}\in\rea^{n_{\tt e}}$. To simplify the notation, we denote the state vector of the EM system as
\begin{equation}
\eta = \begin{bmatrix}
    q \\ p \\ x_{\tt e}
\end{bmatrix}.    
\end{equation}
The dynamics of numerous EM systems can be represented in the pH framework as follows:
\begin{equation} 
\label{eq:open-loop}
\arraycolsep=1pt
\def\arraystretch{1.2}
\begin{array}{rcl}
 \begin{bmatrix}
  \dot{q} \\ \dot{p} \\ \dot{x}_{\tt e}
 \end{bmatrix}&=& \begin{bmatrix}0 & \quad  I & \quad  0 \\ -I & \quad -R_{\tt m} & \quad 0 \\0 & \quad  0 & \quad  J_{\tt e}-R_{\tt e}\end{bmatrix}\begin{bmatrix}
{{\nabla _{{q}}}{H(\eta)}} \\ 
{{\nabla _{{p}}}{H(\eta)}} \\
{{\nabla _{{x}_{\tt e}}}{H(\eta)}}
\end{bmatrix} + \begin{bmatrix}
 0 \\0\\ G_{\tt e}
\end{bmatrix}u,\\[0.7cm]H(\eta)&=& \displaystyle\frac{1}{2}{p^{\top}} M^{-1}p+V(q)
+\frac{1}{2}(x_{\tt e}-\mu(q))^\top \Psi (q)(x_{\tt e}-\mu(q)),  
\end{array}
\end{equation}
where $R_{\tt m}\in\rea^{n_{\tt m} \times n_{\tt m}}$ is the damping matrix of the mechanical subsystem, which is positive definite; $J_{\tt e}\in\rea^{n_{\tt e} \times n_{\tt e}}$ is the interconnection matrix of the electrical subsystem, which is skew-symmetric; $R_{\tt e}\in\rea^{n_{\tt e} \times n_{\tt e}}$ is the damping matrix of the electrical subsystem, which is positive semi-definite; $M \in\rea^{n_{\tt m} \times n_{\tt m}}$ denotes the mass inertia matrix, which is positive definite; $G_{\tt e}\in\rea^{n_{\tt e} \times n_{\tt e}}$ is the full-rank input matrix; $u\in\rea^{n_{\tt e}}$ is the input vector such that the product $G_{\tt e}u$ denotes voltage or current. Moreover, the Hamiltonian $H(\eta)$ describes the total energy of the system. The first term of the Hamiltonian defines the mechanical kinetic energy. The second term i.e., $\MAP{V}{\rea^{n_{\tt m}}}{}$, provides the mechanical potential energy. The last term of the Hamiltonian describes the coupling energy of the system, represented by $\MAP{\mu}{\rea^{n_{\tt m}}}{n_{\tt e}}$, where $\MAP{\Psi}{\rea^{n_{\tt m}}}{n_{\tt e} \times n_{\tt e}}$ represents a capacitance or inductance matrix, and it is positive definite. Note that the third term of the Hamiltonian represents the electrical energy stored in
capacitors or inductors, where the capacitance or inductance depends on the displacement $q$. 
We emphasize that the
possible values for $q$ are determined based on physical
limitations in EM systems.

\subsection{Problem formulation}

The objective of this work is to the trajectory-tracking problem for the class of EM systems characterized by \eqref{eq:open-loop}. To this end, we propose the notion of \textit{contractive pH systems} for a class of EM systems.
Remarkably, contraction ensures that all the trajectories of a system converge exponentially together as time tends to infinity. 
Therefore, we propose a contractive pH system as the target dynamics. However, to guarantee that the trajectories of the closed-loop system converge to the desired one, we need to ensure that the desired trajectory is a feasible one. Hence, to formulate the control problem, we first need to introduce the following definition.

\begin{definition}
    \label{def:feasible}  Given an open set $\mathbb{T} \subset \rea_{+}$, $\eta^\star: \mathbb{T}\rightarrow\rea^{2n_{\tt m}+n_{\tt e}}$ is said to be a feasible trajectory of the system \eqref{eq:open-loop} if there exists $u^\star:\mathbb{T}\rightarrow \rea^{n_{\tt e}}$ such that the pair $(\eta^\star,u^\star)$ is a solution of \eqref{eq:open-loop} for all $t\in \mathbb{T}$.
     \end{definition} 
Given Definition \ref{def:feasible} and the concept of contractive pH systems, the control problem boils down to designing a controller such that the closed-loop system is a contractive pH system and the desired trajectory $\eta^{\star}\triangleq \begin{bmatrix}
    q^\star & p^\star & x^\star_{\tt e}
\end{bmatrix}^\top$ is a feasible trajectory of the closed-loop system.

\section{Control Design Approach}
\label{sec:control}

In this section, a static controller is proposed to track the desired trajectory, $\eta^\star$, of the system \eqref{eq:open-loop}. To this end, we assume that $\eta^\star$ is a feasible trajectory. Moreover, because some results might be local depending on the plant, we define an open subset $D_{\tt 0} \subseteq \Rr^{2n_{\tt m}+n_{\tt e}}$ containing the desired trajectory. 

\noindent
The following assumption characterizes the class of EM considered in this work.
\begin{ass}
\label{ass}
Given \eqref{eq:open-loop}, the potential energy satisfies the following

\begin{equation*}
 \nabla^2 V(q)\succ 0, \quad   \forall \; (q,p,x_{\tt e}) \in D_{\tt 0}.
\end{equation*}
\end{ass}

\begin{remark}
 A broad range of EM systems, especially motors, MEMS, and NEMS, satisfy Assumption \ref{ass}. Intuitively, systems that are not affected by gravity comply with this condition.   
\end{remark}

The following lemma is useful for designing suitable target dynamics---i.e., contractive pH systems---which is the first step to solving the trajectory-tracking problem in the context of this work. In the following lemma, we adapt the results of Theorem 4 in \cite{yaghmaei2017trajectory} to EM systems studied in this paper.

\begin{lemma}
\label{lem}
    Consider the following system
  \begin{equation} 
\label{eq:closed-loop}
\arraycolsep=1pt
\def\arraystretch{1.2}
\begin{array}{rcl}
 \begin{bmatrix}
  \dot{q} \\ \dot{p} \\ \dot{x}_{\tt e}
 \end{bmatrix}&=& \begin{bmatrix}0 & \quad  I & \quad  0 \\ -I & \quad -R_{\tt m} & \quad 0 \\0 & \quad  0 & \quad   J_{\tt e}-\bar R_{\tt e}\end{bmatrix}\begin{bmatrix}
{{\nabla _{{q}}}{H_d(\eta,t)}} \\ 
{{\nabla _{{p}}}{H_d(\eta,t)}} \\
{{\nabla _{{x}_{\tt e}}}{H_d(\eta,t)}}
\end{bmatrix},\\[0.7cm]H_d(\eta,t)&=& H(\eta)+\Theta(x_{\tt e},t),
\end{array}
\end{equation} 
where $\bar R_{\tt e}$ is positive definite and  $\MAP{\Theta}{\rea^{n_{\tt e}} \times \rea_{+}}{}$ is a twice differentiable function. Suppose that  $\Theta$ and $\bar R_{\tt{e}}$ can be chosen such that
\begin{itemize}
    \item [$\bullet$] The following inequality holds
	 \begin{equation}
	 \label{eq:cond2}
	 \beta_1 I\prec \nabla_\eta^2 H(\eta)+\nabla_\eta^2 \Theta (x_{\tt e}, t) \prec \beta_2 I, \quad \eta \in D_0
	 \end{equation}
   where $\beta_1,\beta_2>0$.
    \item [$\bullet$] The matrix  
	\begin{equation}
	\label{eq:cond3}
	N= \begin{bmatrix}
	P & \left(  1-\frac{\beta_1} {\beta_2} \right) PP^\top\\ -( 1-\frac{\beta_1}{\beta_2}+\varepsilon)I & -P^\top
	\end{bmatrix},
	\end{equation}
    where
    \begin{equation}
	 \label{eq:cond1}
	 P\triangleq \begin{bmatrix}
	 0 & \quad  I & \quad  0 \\ -I & \quad -R_{\tt m} & \quad 0 \\
	 0 & \quad  0 & \quad  J_{\tt e} - \bar{R}_{\tt e}\end{bmatrix},
	 \end{equation} 
    has no eigenvalues on the imaginary axis for a positive constant $\varepsilon$.
    \end{itemize}
    
    Hence, the system \eqref{eq:closed-loop} is contractive on the open subset $D_0$.
    \end{lemma}
    \begin{proof}
    It follows from \cite[Theorem 4]{yaghmaei2017trajectory} that
    $P$ is designed as a Hurwitz matrix. For this purpose, we redefine $P$ in  \eqref{eq:cond1} as follows
  \begin{equation}
	 \label{eq:bar-P}
	 P= \begin{bmatrix}
	 F_1 & \quad  0_{2n_{\tt m} \times n_{\tt e}}  \\[2 mm]
	 0_{n_{\tt e} \times 2n_{\tt m}} & F_2 \end{bmatrix},
	 \end{equation} 
  where 
\begin{equation}
F_1 \triangleq \begin{bmatrix}
0_{n_{\tt m} \times n_{\tt m}} & \quad  I\\[2 mm]
-I & \quad -R_{\tt m}
\end{bmatrix},
\end{equation}
and $F_2 \triangleq J_{\tt e} - \bar{R}_{\tt e}$. Therefore, $P$ is Hurwitz if and only if $F_1$ and $ F_2$ are Hurwitz. Since $ F_2 \prec 0 $, this matrix is Hurwitz. To prove $F_1$ is Hurwitz, consider
    \begin{equation*}
      S(\lambda)\triangleq\lambda I-F_1.
  \end{equation*}
Therefore, the eigenvalues of $F_1$ are determined by the following equation:
  \begin{equation*}
      \textnormal{eig}(F_1)=\det\big( S(\lambda)\big)=0.
  \end{equation*}
According to the Schur complement, we have
  \begin{equation*}
      \textnormal{eig}(F_1)=\det \left(I\lambda\right){\det \left(I\lambda+R_{\tt m}+I\frac{1}{\lambda}\right)}=
      \det\left(I\lambda \frac{1}{\lambda}\right){\det (I\lambda^2+R_{\tt m}\lambda+I)}=0.
  \end{equation*}
Furthermore, there exists an orthogonal matrix $U_{\tt m}$ and a positive definite matrix $\Sigma_{\tt m} \triangleq \mbox{diag}(\sigma_1,...,\sigma_{{n_{\tt m}}}) $, where $\sigma_i>0$, $i\kern-.3em\in\kern-.3em \{1,...,n_{\tt m}\}$, such that $\Sigma_{\tt m}=U_{\tt m}^\top R_{\tt m} U_{\tt m}$.
Therefore,  
\begin{equation}
\label{eq:detL}
\textnormal{eig}(F_1)={\det (I\lambda^2+R_{\tt m}\lambda+I)}=
\det(\lambda^{2} I+\Sigma_{\tt m}\lambda+I)=\prod_{\tt i=1}^{n_{\tt m}}(\lambda^2+\sigma_i\lambda+1).
\end{equation}
Because $\sigma_i>0$, the eigenvalues of \eqref{eq:detL} are negative, implying that $F_1$ is Hurwitz. Accordingly, $P$ defined in \eqref{eq:cond1} is Hurwitz. Besides, conditions \eqref{eq:cond2} and \eqref{eq:cond3}  follow immediately from the results reported in \cite[Theorem 4]{yaghmaei2017trajectory}. Hence, \eqref{eq:closed-loop} is contractive. $\hfill \square$
    \end{proof}

The following proposition provides a static tracking controller for the EM systems characterized by \eqref{eq:open-loop}. To this end, based on Lemma \ref{lem}, we propose a contractive closed-loop system such that $\eta^\star$ is a feasible trajectory of the closed-loop system. Hence, the
convergence property of the contractive systems ensures that all the trajectories of the closed-loop system converge exponentially over time, solving the trajectory-tracking problem.

\begin{proposition}
Consider the EM system \eqref{eq:open-loop} satisfying Assumption \ref{ass} and the feasible trajectory $\eta^\star$.
Suppose that there exist positive definite matrices $\bar R_{\tt e}$ and $K_{\tt c}$ satisfying the conditions of Lemma \ref{lem} for $\Theta=\frac{1}{2}(x_{\tt e}-\alpha(t))^\top K_{\tt c}(x_{\tt e}-\alpha(t))$, where
\begin{align}
   \alpha(t)\triangleq{K_{\tt c}}^{-1}{\Psi}(q^\star)(x^\star_{\tt e}-\mu(q^\star))+x_{\tt e}^\star- K_{\tt c}^{-1} (J_{\tt e} -\bar R_{\tt e})^{-1}\dot{x}_{\tt e}^\star.\label{controller}
\end{align}

The static feedback controller 
\begin{align}
\label{eq:controller}
	u=G_{\tt e} ^{\dagger} \bigg(( R_{\tt e}-\bar R_{\tt e}){\Psi(q)(x_{\tt e}-\mu(q))}+(J_{\tt e}-\bar R_{\tt e})K_{\tt c}(x_{\tt e}-\alpha(t))\bigg)
\end{align}
guarantees that the trajectories of the closed-loop system converge exponentially to $\eta^{\star}$. 
\end{proposition}
\begin{proof}
Applying the controller \eqref{controller}, the closed-loop system reduces to \eqref{eq:closed-loop}. Because the conditions of Lemma \ref{lem}  are satisfied by $\bar R_{\tt e}$ and $K_{\tt c}$, the closed-loop system \eqref{eq:closed-loop} is contractive. In addition, $\alpha(t)$ in \eqref{controller} ensures that $\eta^\star(t)$ is a feasible trajectory of the closed-loop system. Thereby, all trajectories of the closed-loop system exponentially converge to $\eta^\star(t)$ due to the convergence property in contractive systems.$\hfill \square$
\end{proof}

\begin{remark}
 The controller \eqref{controller} solves the regulation problem if $\eta^\star$ is chosen constant. Moreover, in such a case, $\eta^\star$ is an exponentially stable equilibrium point for the closed-loop system.  
\end{remark}

\section{Examples}
\label{sec:ex}

In this section, we apply the results from Section \ref{sec:control} to various EM systems, indicating the performance of the proposed controller \eqref{eq:controller}.
For this purpose, we select a stepper motor system with invariant $\mu(q)$, microphone, and loudspeaker. In particular, the microphone system stores the electrical energy in a capacitor, while the loudspeaker system utilizes an inductor for the electrical energy storage. It is worth noting that we have shown input control plots with a starting point of 0.001 seconds. The reason is that there are spikes at the starting point due to numerical errors in MATLAB. Therefore, we remove the first points of the plots to clearly indicate the magnitude range thereafter. 

\subsection{Stepper Motor}

 The EM pH system \eqref{eq:open-loop} can effectively model a 2$\phi$ stepper motor. In this presentation, $q$, $p$, and $x_{\tt e}$ correspond to the mechanical angular position, mechanical angular momentum, and the stator flux linkage, respectively. Furthermore, $n_{\tt m}=1, n_{\tt e}=2$, $J_{\tt e}=0_{2 \times 2}$, $R_{\tt e}=R_{\tt s} I_{2 \times 2}$, $G_{\tt e}=2 I_{2 \times 2}$, $V(q)=-\frac{1}{4N_{\tt r}}k_{\tt D}\cos(4N_{\tt r}q)+\tau_{\tt L} q$,  $\mu(q)=\frac{k_{\tt m}}{N_{\tt r}} \begin{bmatrix}\cos(N_{\tt r}q),\sin(N_{\tt r}q)\end{bmatrix}^\top$, $\Psi(q)=L_{\tt s} I_{2\times 2}$,
 where $L_{\tt s}$ is the stator winding inductance, $R_{\tt s}$ is the stator winding resistance, $k_{\tt D}$ is the \textit{detente torque} constant, $N_{\tt r}$ denotes the number of teeth of the rotor, $k_{\tt m}$
 is the \textit{dq} back emf constant, $\tau_{\tt L}$ is the electromechanical torque, and $m>0$. The tracking problem requires $q$ to track the reference signal $f(t)$. Therefore, the feasible trajectory $\eta$ is obtained as follows 
\begin{equation}
\label{reference trajectory}
\begin{split}
  q^\star=f, \  p^\star=m\dot{f}, \ x^{\star}_{{\tt e}_1}= -\frac{L_{\tt s}C}{k_{\tt m}}\sin({N_{\tt r}f}),
  \ x^{\star}_{{\tt e}_2}= \frac{L_{\tt s}C}{k_{\tt m}}\cos({N_{\tt r}f}),
\end{split}    
\end{equation}
  where $q^\star \in \begin{bmatrix}-\pi,\pi\end{bmatrix}$ and $C=m\ddot{f}+k_{\tt D}\sin(4N_{\tt r}f)+R_{\tt m} \dot{f}+\tau_{\tt L}$. Besides, the subsequent numerical parameters are used

\begin{equation*}
\begin{split}
   & f(t)=0.5\sin(t), \quad M=1.872 \times 10^{-4}, \quad R_{\tt m}=2, \quad R_{\tt s}=0.9, \quad k_{\tt m}=2.25,\\& k_{\tt D}=0.0176, \quad \tau_{\tt L}=1.720129, \quad N_{\tt r}=6,\quad
     \bar{R}_{\tt e}=\mbox{diag}(23,25), \quad K_{\tt c}=35.
\end{split}
\end{equation*} 
and the initial conditions of the states are 
\begin{equation*}
\eta_0(t)=\begin{bmatrix}
    0.2,-0.02,1.6,-0.3
\end{bmatrix}^\top.
\end{equation*}
 
As shown in  Fig. \ref{fig:motorE}, the state error signals converge to zero over time. Besides, the control signals are depicted in Fig. \ref{fig:motorE}. Furthermore, Fig. \ref{fig:motorT} shows the trajectory convergence to the desired reference signals.

\begin{figure}[!ht]
\includegraphics[width=1.05\columnwidth]{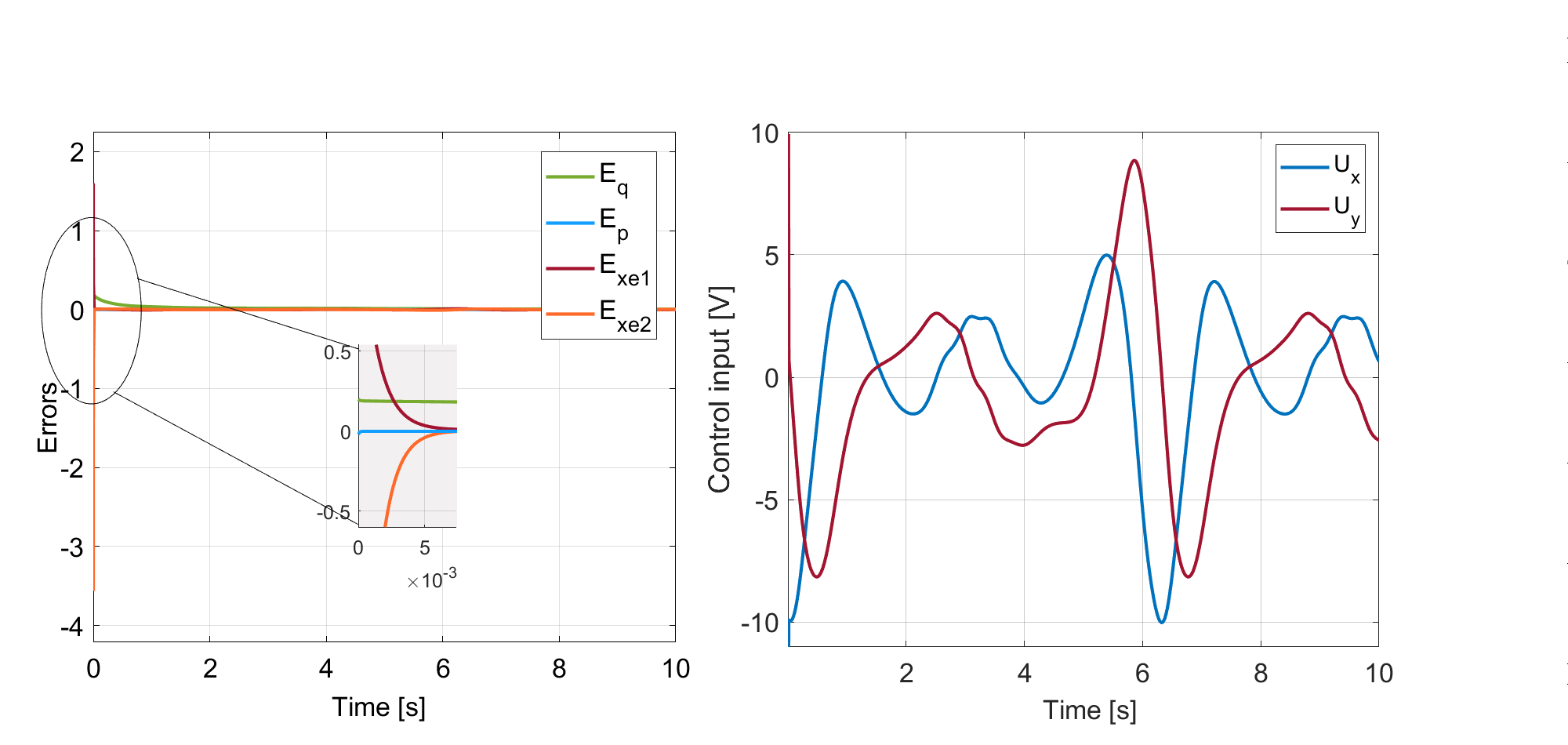}
\centering
\vspace{-0.3cm}
\caption{The position, velocity, and electrical error signals, which are denoted by $E_{\tt q}$, $E_{\tt p}$, and $E_{\tt x_e}$, respectively. Besides, the control signals from the voltage resource are denoted by $U_{\tt x}$ and $U_{\tt y}$, respectively. }
\label{fig:motorE}
\end{figure}

\begin{figure}[!ht]
\includegraphics[width=1\columnwidth]{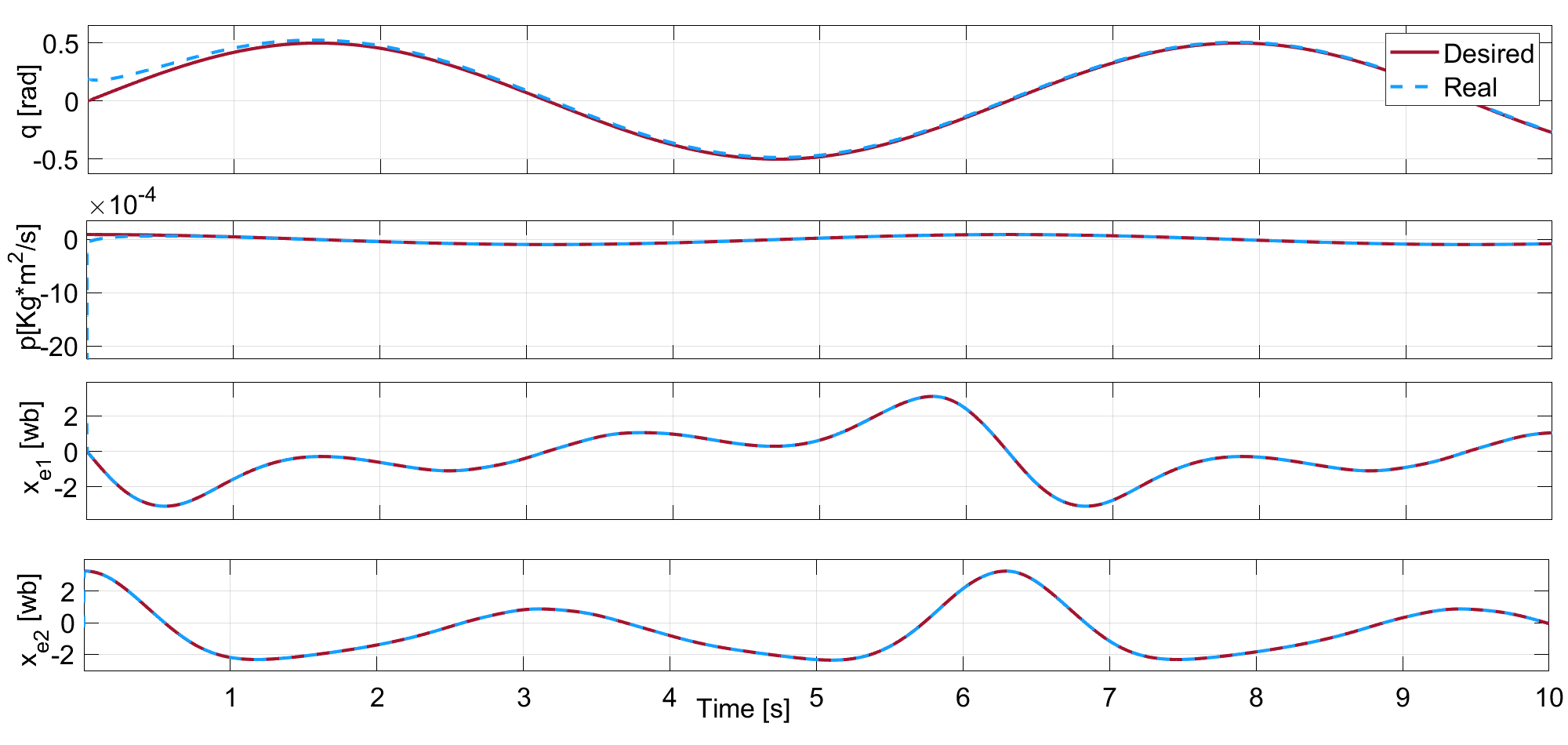}
\centering
\vspace{-0.3cm}
\caption{The desired and closed-loop trajectories for the stepper motor. Note that $q(t)$ tracks the reference signal $f(t)=0.5\sin(t)$.}
\label{fig:motorT}
\end{figure}

\subsection {Microphone system}

The dynamics of a microphone can be described by the EM pH model \eqref{eq:open-loop}, where $q$, $p$, and $x_{\tt e}$ represent the displacement of the movable capacitor plate, its momentum, and the charge at the capacitor, respectively. 
Besides, the model parameters are considered as
$n_{\tt m}=n_{\tt e}=1$, $J_{\tt e}=0$, $G_{\tt e}=1$, $V(q)=\frac{1}{2}k(q-\gamma_1)^2$, $\mu(q)=0$, $\Psi(q)=\frac{1}{C(q)}$, where 
$k$ and $\gamma_1$ are positive. Similarly, $q\in(0,\infty)$. Hence, the capacitance $C(q)>0$ satisfies $C(q)=\frac{1}{q}$.

For the tracking trajectory problem, we want $q$ to track the reference signal $f(t)$. Therefore, the feasible desired trajectories are 
\begin{equation}
\label{reference trajectory}
\begin{split}
  q^\star=f, \  p^\star=m\dot{f}, \ x^{\star}_{\tt e}= \sqrt{2k(\gamma_1-f)-2R_{\tt m}\dot{f}-2m\ddot{f}}, 
\end{split}    
\end{equation}
where $2k(\gamma_1-f)>2R_{\tt m}\dot{f}+2m\ddot{f}$. For simulation purposes, the system and control parameters are chosen as follows:

\begin{equation*}
\begin{split}
   &f(t)=0.3+0.2\sin(t), \quad M=1, \quad R_{\tt m}=1, \quad R_{\tt e}=1, \quad k=1, \quad \gamma_1=1,\\&
     \bar{R}_{\tt e}=30, \quad K_{\tt c}=55.
\end{split}
\end{equation*}
and the initial conditions of the states are 

\begin{equation*}
\eta_0(t)=\begin{bmatrix}
   0.02,-0.02,0.3
\end{bmatrix}^\top.
\end{equation*}

The results are shown in Figs. \ref{fig:micerror} and \ref{fig:mictrajectory} . Besides, the position, velocity, and coupling errors of the system converge to zero, as shown in Fig. \ref{fig:micerror}.   
Note that the system states exponentially track the desired references \eqref{reference trajectory} in Fig. \ref{fig:mictrajectory}

\begin{figure}[!ht]
\includegraphics[width=1.05\columnwidth]{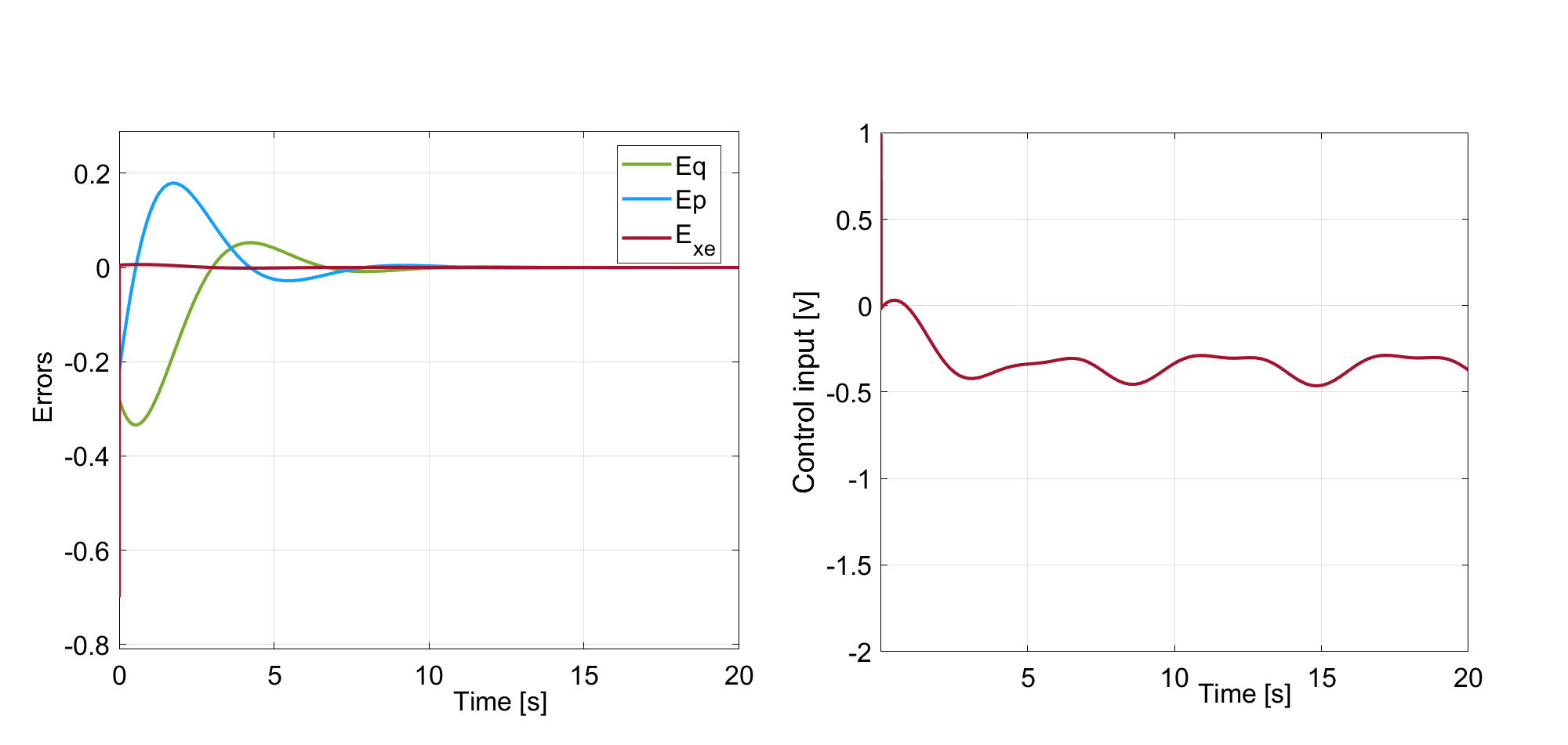}
\centering
\vspace{-0.3cm}
\caption{The position, velocity, and coupling error signals, which are denoted by $E_{\tt q}$, $E_{\tt p}$ and $E_{\tt x_e}$, and the control signal from a voltage resource, respectively. }
\label{fig:micerror}
\end{figure}
\begin{figure}[!ht]
\includegraphics[width=1\columnwidth]{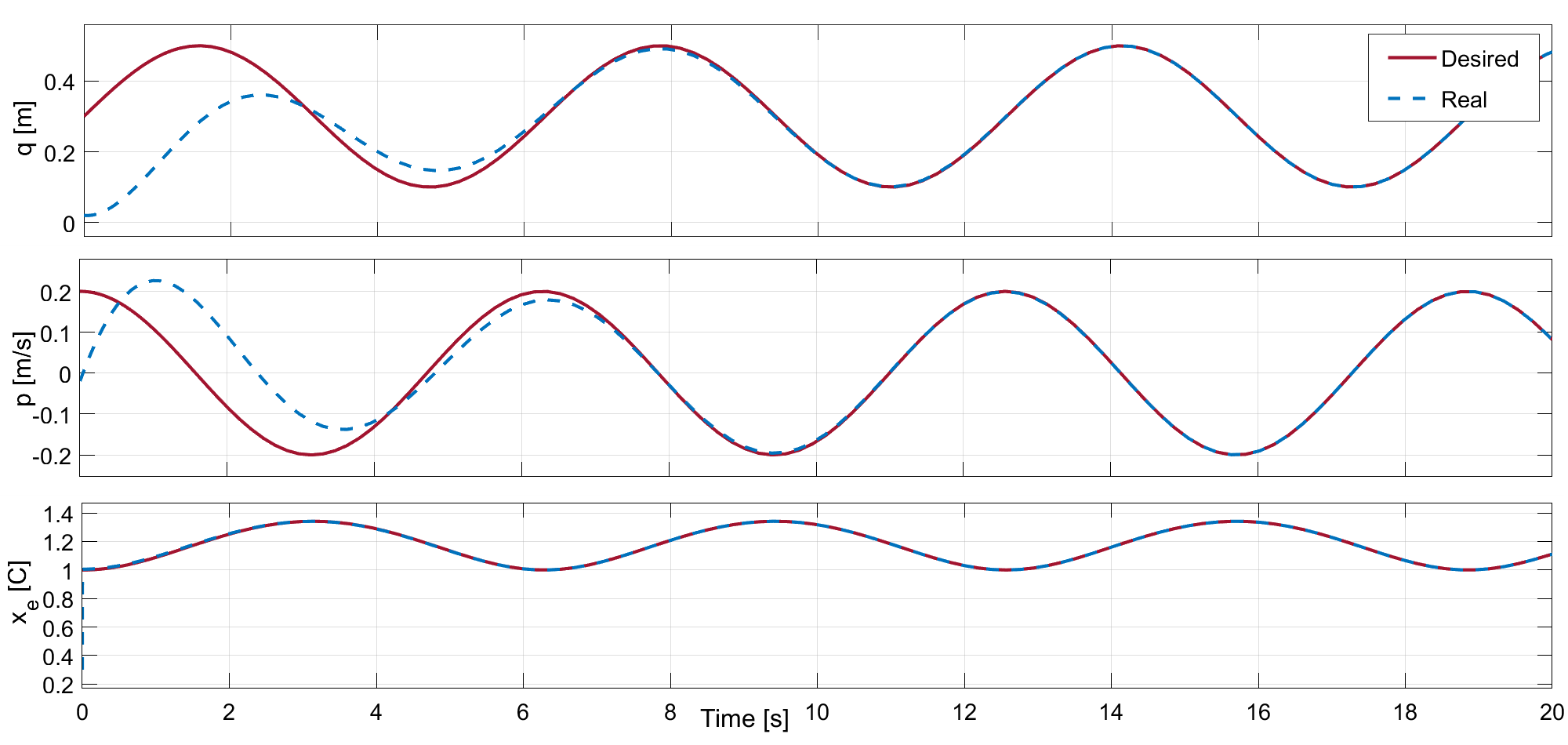}
\centering
\vspace{-0.3cm}
\caption{ The desired and closed-loop trajectories for the microphone system. Note that $q(t)$ tracks the signal $f(t)=0.3+0.2\sin(t)$. }
\label{fig:mictrajectory}
\end{figure}

\subsection{Loudspeaker system}
The dynamics of a loudspeaker can be described by the EM pH model \eqref{eq:open-loop} with $n_{\tt m}=n_{\tt e}=1$.
In this model, $q$, $p$, and $x_{\tt e}$ denote the mechanical displacement of the voice coil, the translational momentum, and the current through the inductor over the displacement, respectively. Furthermore, the model parameters are $R_{\tt e}=0$, $J_{\tt e}=0$, $G_{\tt e}=1$, $V(q)=\frac{1}{2}kq^2$, $\mu(q)=0$, $\Psi(q)=\frac{1}{\alpha L(q)}$, where $k$ and $\alpha$ are positive. Similarly,
$q \in (0, \infty)$. The inductance $L(q)>0$ satisfies $L(q)=\frac{1}{q}$. 

For the tracking trajectory aim, we consider that $q$ tracks a reference signal $f(t)$. The feasible desired trajectories follow by the same formulation in \eqref{reference trajectory} with $\gamma_1=0$. The control and system parameters are chosen as follows:
\begin{equation*}
\begin{split}
   & f(t)=0.3\sin(t)-1, \quad M=1, \quad R_{\tt m}=1, \quad R_{\tt e}=1, \quad k=1, \quad  \alpha=4\\&
     \bar{R}_{\tt e}=100, \quad K_{\tt c}=3.5.
\end{split}
\end{equation*}
and the initial conditions of the states are 

\begin{equation*}
\eta_0(t)=\begin{bmatrix}
   0.02,-0.02,2.3
\end{bmatrix}^\top.
\end{equation*}

 Hence, Fig. \ref{fig:lserror} shows that the state errors converge to zero within the time. Besides, the control signal can be evaluated by Fig. \ref{fig:lserror}. 
 In particular, the mechanical and electrical states follow the mentioned desired reference trajectory in Fig. \ref{fig:lstrajectort}.  
\begin{figure}[!ht]
\includegraphics[width=1\columnwidth]{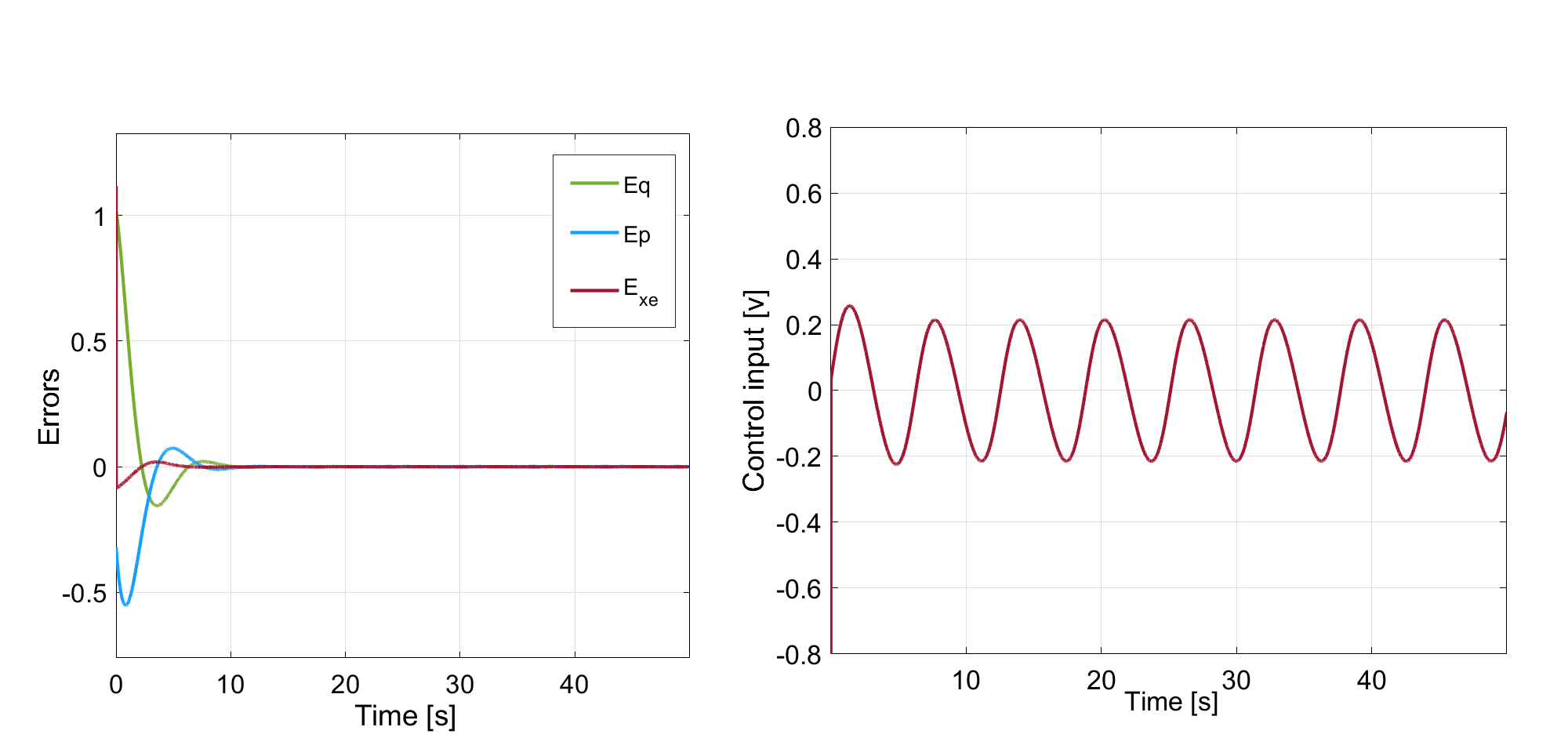}
\centering
\vspace{-0.3cm}
\caption{The position, velocity, and coupling error signals, which are denoted by $E_{\tt q}$, $E_{\tt p}$, and $E_{\tt x_e}$, and the control signal from a voltage resource, respectively. }
\label{fig:lserror}

\end{figure}

\begin{figure}[!ht]
\includegraphics[width=1\columnwidth]{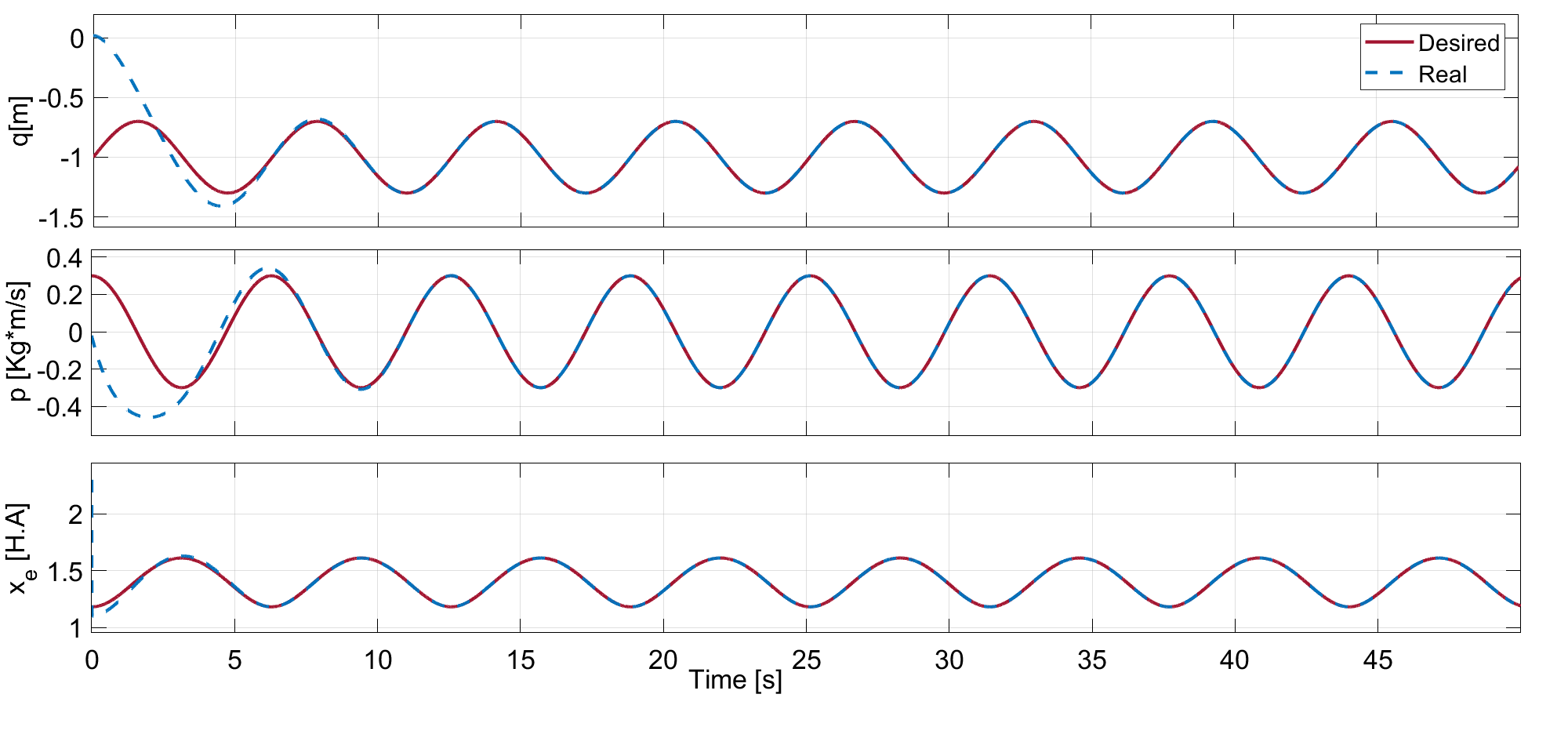}
\centering
\vspace{-0.3cm}
\caption{ The desired and closed-loop trajectories for the loudspeaker system. Note that $q(t)$ tracks $f(t)=0.3\sin t-1$. }
\label{fig:lstrajectort}
\end{figure}
\section{Concluding Remarks}
\label{sec:con}
Contraction is a powerful tool that can be exploited for trajectory-tracking purposes. In this paper, we have shown how the trajectory problem for a class of systems can be solved by proposing a contractive closed-loop system while guaranteeing the desired trajectory is feasible. Moreover, the target dynamics are designed constructively so that no PDEs need to be solved to formulate the controller that achieves the desired closed-loop system.

\section*{Declarations}
\subsection*{Availability of data and material}
The data associated with the simulations are available upon reasonable request.
\subsection*{Competing interests}
The authors declare that they have no conflict of interest concerning the publication of this manuscript.
\subsection*{Funding}
The authors didn't receive particular funding for this work.
\subsection*{Authors' contributions}
{\bf Najmeh Javanmardi.} was in charge of the technical developments, simulations, and writing of this manuscript.\\[0.1cm]
{\bf Pablo Borja.} co-supervised the research process, helped with the technical developments, and suggested some of the chosen research directions explored in this work.\\[0.1cm]
{\bf Jacquelien M. A. Scherpen.} provided her expert guidance throughout the research process and spearheaded the conception of the study.
\subsection*{Acknowledgements}
The authors want to thank the French-Dutch-German Doctoral College of Port-Hamiltonian Systems Modeling Numerics and Control for the valuable feedback received during numerous interactions.

\end{document}